%% file: sea2010_arXiv.tex
\documentclass{llncs}

\usepackage{amsmath, amssymb}

\usepackage{chngpage}
\usepackage{xspace}
\usepackage[usenames]{color}
\usepackage[latin1]{inputenc}
\usepackage[american]{babel}
\usepackage{graphicx}
\usepackage{epsfig}
\usepackage{color}
\usepackage{times}
\usepackage{url}
\usepackage{multirow}
\usepackage[small,bf]{caption} 

\usepackage{setspace}

\usepackage{subfigure}

\usepackage[numbers,sort&compress,longnamesfirst,sectionbib]{natbib}             

\newcommand{\PP}{{\ensuremath{\mathcal{P}}}}

\newcommand{\IGNOREME}[1]{}

\numberwithin{thm}{section}

\DeclareSymbolFont{AMSb}{U}{msb}{m}{n}
\newcommand{\N}{{\mathbb{N}}}

\newcommand{\R}{{\mathbb{R}}}

\newcount\minute \newcount\hour \newcount\hourMins
\def\now{\minute=\time \hour=\time \divide \hour by 60 \hourMins=\hour \multiply\hourMins by 60
  \advance\minute by -\hourMins \zeroPadTwo{\the\hour}:\zeroPadTwo{\the\minute}}

\def\today{\the\year-\zeroPadTwo{\the\month}-\zeroPadTwo{\the\day}}
\def\zeroPadTwo#1{\ifnum #1<10 0\fi #1}
\makeatletter
\DeclareRobustCommand{\qed}{\ifmmode\mathqed\else\leavevmode\unskip\penalty9999\hbox{}\nobreak\hfill\quad\hbox{\qedsymbol}\fi}
\let\QED@stack\@empty
\let\qed@elt\relax
\newcommand{\pushQED}[1]{\toks@{\qed@elt{#1}}\@temptokena\expandafter{\QED@stack}\xdef\QED@stack{\the\toks@\the\@temptokena}}
\newcommand{\popQED}{\begingroup\let\qed@elt\popQED@elt \QED@stack\relax\relax\endgroup}
\def\popQED@elt#1#2\relax{#1\gdef\QED@stack{#2}}
\newcommand{\qedhere}{\begingroup \let\mathqed\math@qedhere\let\qed@elt\setQED@elt \QED@stack\relax\relax \endgroup}
\newif\ifmeasuring@
\newif\iffirstchoice@ \firstchoice@true
\def\setQED@elt#1#2\relax{\ifmeasuring@\else \iffirstchoice@ \gdef\QED@stack{\qed@elt{}#2}\fi\fi#1}
\newcommand{\mathqed}{\quad\hbox{\qedsymbol}}
\def\linebox@qed{\hfil\hbox{\qedsymbol}\hfilneg}
\def\math@qedhere{\@ifundefined{\@currenvir @qed}{\qed@warning\quad\hbox{\qedsymbol}}{\@xp\aftergroup\csname\@currenvir @qed\endcsname}}
\def\displaymath@qed{\relax\ifmmode\ifinner\aftergroup\linebox@qed\else\eqno\let\eqno\relax \let\leqno\relax \let\veqno\relax\hbox{\qedsymbol}\fi\else\aftergroup\linebox@qed\fi}
\@xp\let\csname equation*@qed\endcsname\displaymath@qed
\def\equation@qed{
  \iftagsleft@\hbox{\phantom{\quad\qedsymbol}}\gdef\alt@tag{\rlap{\hbox to\displaywidth{\hfil\qedsymbol}}\global\let\alt@tag\@empty}
  \else\gdef\alt@tag{\global\let\alt@tag\@empty\vtop{\ialign{\hfil####\cr\tagform@\theequation\cr\qedsymbol\cr}}\setbox\z@}
  \fi
}
\def\qed@tag{\global\tag@true \nonumber&\omit\setboxz@h {\strut@ \qedsymbol}\tagsleft@false\place@tag@gather\kern-\tabskip\ifst@rred \else \global\@eqnswtrue \fi \global\advance\row@\@ne \cr}
\def\split@qed{\def\endsplit{\crcr\egroup \egroup \ctagsplit@false \rendsplit@\aftergroup\align@qed}}
\def\align@qed{\ifmeasuring@ \tag*{\qedsymbol}\else \let\math@cr@@@\qed@tag\fi}
\@xp\let\csname align*@qed\endcsname\align@qed
\@xp\let\csname gather*@qed\endcsname\align@qed
\def\@tempb#1 v#2.#3\@nil{#2}
\ifnum\@xp\@xp\@xp\@tempb\csname ver@amsmath.sty\endcsname v0.0\@nil<\tw@\def\@tempa{TT}\else\def\@tempa{TF}\fi
\if\@tempa\renewcommand{\math@qedhere}{\quad\hbox{\qedsymbol}}\fi
\newcommand{\openbox}{\leavevmode\hbox to.77778em{\hfil\vrule\vbox to.675em{\hrule width.6em\vfil\hrule}\vrule\hfil}}
\DeclareRobustCommand{\textsquare}{\begingroup\usefont{U}{msa}{m}{n}\thr@@\endgroup}
\providecommand{\qedsymbol}{\openbox}
\renewenvironment{proof}[1][\proofname]{\par\pushQED{\qed}\normalfont\topsep6\p@\@plus6\p@\relax\trivlist\item[\hskip\labelsep\itshape #1\@addpunct{.}]\ignorespaces}{\popQED\endtrivlist\@endpefalse}
\providecommand{\proofname}{Proof}
\makeatother

\title{Randomized Rounding for Routing and Covering Problems: Experiments and Improvements\thanks{Supported by the German Science Foundation (DFG) via its priority program (SPP) 1307 ``Algorithm Engineering'', grant DO 749/4-2.}}

\author{Benjamin Doerr\inst{1},
        Marvin K\"unnemann\inst{2}, \and
        Magnus Wahlstr\"om\inst{1}}

\institute{
    Max-Planck-Institut f\"ur Informatik, Saarbr\"ucken, Germany
    \and
    Universit\"at des Saarlandes, Saarbr\"ucken, Germany}

\begin{document}

\maketitle


\begin{abstract} 
  Following previous theoretical work by Srinivasan (FOCS 2001) and the first author (STACS 2006) and a first experimental evaluation on random instances (ALENEX 2009), we investigate how the recently developed different approaches to generate randomized roundings satisfying disjoint cardinality constraints behave when used in two classical algorithmic problems, namely low-congestion routing in networks and max-coverage problems in hypergraphs.  
  
  We generally find that all randomized rounding algorithms work well, much better than what is guaranteed by existing theoretical work.  The derandomized versions produce again significantly better rounding errors, with running times still negligible compared to the one for solving the corresponding LP. It thus seems worth preferring them over the randomized variants.
  
  The data created in these experiments lets us propose and investigate the following new ideas. For the low-congestion routing problems, we suggest to solve a second LP, which yields the same congestion, but aims at producing a solution that is easier to round. Experiments show that this reduces the rounding errors considerably, both in combination with randomized and derandomized rounding.
  
For the max-coverage instances, we generally observe that the greedy heuristics also performs very good. We develop a strengthened method of derandomized rounding, and a simple greedy/rounding hybrid approach using greedy and LP-based rounding elements, and observe that both these improvements yield again better solutions than both earlier approaches on their own. 

  
  For an important special case of max-coverage, namely unit disk max-domination, we also develop a PTAS. Contrary to all other algorithms investigated, it performs not much better in experiments than in theory. In consequence, unless extremely good solutions are to be obtained with huge computational resources, greedy, LP-based rounding or hybrid approaches are preferable.
\end{abstract}


\section{Introduction}

Randomized rounding  is one of the core primitives in randomized algorithmics. In contrast to many deep theoretical results, only very little experimental knowledge exists, and almost no fine-tuning and other implementation advice exists. Such results became even more interesting, since in the last ten years two substantially different methods~\cite{sriniround,gandhifocs02,gandhijacm,ichstacs06} extending the classical approach of Raghavan and Thompson~\cite{raghthom,ragh} were developed.

The only experimental work on either classical randomized rounding or the new approaches seems to be~\cite{DW09}. It compares the different methods on randomly generated rounding problems. The purpose of this work is to extend these results to two less artificial problem classes, namely routing and covering problems. These problems are among the first ones for which randomized rounding has been proven (by theoretical means) to lead to good algorithms.

\noindent{\bf Randomized Rounding: }
Given an arbitrary real number $x$, we say that (the random variable) $y$ is a randomized rounding of $x$, if $y$ equals $\lfloor x\rfloor+1$ with probability $\{x\} := x - \lfloor x \rfloor$ and $\lfloor x \rfloor$ otherwise. In simple words, the closer $x$ is to the next larger integer, the higher the chance of being rounded up. 

Randomized rounding builds on the simple observation that this keeps the expectation unchanged, that is, $E(y) = x$. This naturally extends to linear expressions. If $y_1, \ldots, y_n$ are randomized roundings of $x_1, \ldots, x_n$ and $f : \R^n \to \R$ is a linear function, then $E(f(y_1, \ldots, y_n)) = f(x_1, \ldots, x_n)$. If, in addition, the $y_i$ are independent, then Chernoff bounds allow strong quantitative statements showing that with high probability, $f(y_1, \ldots, y_n)$ is not far from its expectation. These two key facts allow to use randomized rounding in connection with integer linear programming. Two examples of this are given in the following sections. 

The new aspect of the works~\cite{sriniround,gandhifocs02,gandhijacm,ichstacs06} is that they allow to generate randomized roundings that satisfy certain cardinality constraints with probability one. That is, for certain sets $I$, we can prescribe that $\sum_{i \in I} y_i = \sum_{i \in I} x_i$, provided the right-hand side is integral. This can be done without giving in with the other properties---both methods generate randomized roundings that admit the same Chernoff bounds as independent randomized rounding. 

For reasons of space, we cannot describe these rounding algorithms here. However, since in this paper we are mainly comparing them experimentally, the reader may treat them as black-box, keeping in mind only that they generate randomized roundings that look independent, but satisfy cardinality constraints. 

We shall concentrate on disjoint cardinality constraints. This is the most common form of cardinality constraints. Also, the comparison of the methods is more interesting here, since all have the same time complexity. 

\noindent{\bf Our Results: }
The aim of this work is to find out how well the different rounding approaches are suited to solve classical problems that are often attacked with LP-based methods, but also to try to find fine-tunings and alternative approaches. 

As underlying problems we chose the classical low-congestion routing problem and the max-coverage problem. They are different in flavor since in the first, randomized rounding is used with a focus of exploiting Chernoff bounds in linear constraints and objective function. In the second, since the right-hand side of the inequalities is one, we cannot do so, but resort to accepting that a certain fraction of the vertices covered in the relaxation are not covered after rounding.

All our results indicate that generally the derandomized algorithms yield superior results. The increase in running time over the randomized versions usually is still negligible compared to the complexity of solving the LPs involved.

For the low-congestion routing problem, we regard routing requests placed randomly on a two-dimensional grid. We regard instances small enough to compute the optimum solutions via solving an integer linear program. We observe that randomized rounding with cardinality constraints obtains reasonably good solutions. Surprisingly, unlike in previous experiments~\cite{DW09}, we observe that the bit-wise randomized approach of~\cite{ichstacs06} produces better results than the tree-based one of~\cite{sriniround}.

The gap to the optimum is roughly halved if we use a derandomization of randomized rounding. Here, the derandomization of~\cite{sriniround} obtained in~\cite{DW09} proved to be superior.

In an attempt to fine-tune randomized rounding, we propose solving a second LP which gives the same congestion in the relaxation, but aims at making the solution easier to round. While this naturally does not give improved theoretical guarantees, it yields a good reduction of the rounding errors, in particular, in combination with the derandomization. This seems to be a fruitful approach whenever the additional cost of solving a second LP is admissible.

Our analysis of maximum coverage shows that both randomized rounding and the greedy algorithm produce good results in general. However, for both there are instances showing the other behave much better. Analyzing the data produced in our experiments, we consider two paths to a hybrid approach.  One way is to strengthen the derandomization to include a greedy component in variable selection, as a gradient-based rounding; the other, complementary, is to spend part of the budget greedily, and solve the remaining instance via an LP- and randomized rounding-approach.  Both hybrids perform better than either of the two plain approaches alone; the gradient-based rounding performs particularly well.

For a natural planar Euclidian version of the problem, we also give a PTAS. However, unlike for all other approaches used in this paper, the experimental results are not much better than the theoretical guarantees. In consequence, this is an alternative useful only if very good approximations are needed and if computation power is available plentiful.

\section{Randomized Rounding for Low-Congestion Routing}
\label{sec:routing}

\subsection{The Low-Congestion Routing Problem in Networks}

The low-congestion routing problem is one of the classical applications of randomized rounding~\cite{raghthom}. In its simplest version, the objective is to route a number of requests through a given network, minimizing the maximum usage of an edge (``congestion''). Problems of this type found all kinds of applications, an early one being routing wires in gate arrays~\cite{karprouting}.

We shall regard the following basic variant, previously regarded also in~\cite{raghthom, ragh, kleinbergthesis, sriniround, gandhijacm}. Given is a (directed) network $G=(V,E)$, together with $k$ routing requests. Each consists of a source vertex $s_i$, a target vertex $t_i$ and a demand $r_i$. The objective is to find, for each $i \in [k] := \{1, \dots, k\}$, a flow from $s_i$ to $t_i$ having flow value $r_i \in \N$, such that the congestion, that is, the maximum total flow over an edge, is minimized. This problem is easily formulated as integer linear program (ILP): We minimize the congestion $C$ subject to the constraints 
\begin{align} 
  &\forall e \in E: \sum_{i = 1}^k x_{ie} \le C \label{ilpcongestion}\\
  &\forall i \in [k]: \sum_{e = (s_i,v) \in E} x_{ie} - \sum_{e = (v,s_i) \in E} x_{ie} = r_i \label{ilpcard}\\
  &\forall i \in [k] \, \forall v \in V \setminus \{s_i,t_i\}: \sum_{e = (w,v) \in E} x_{ie} = \sum_{e = (v,w) \in E} x_{ie}\\
  &\forall i \in [k] \, \forall e \in E: x_{ie} \in \{0, 1\}.\label{ilplast}
\end{align}
We should add that~\cite{raghthom} only regard the special case of all demands $r_i$ being one, since randomized rounding respecting cardinality constraints with right-hand side greater than one was not available at that time. For an application with particular need for larger $r_i$, see the failure restoration problem in optical networks described in~\cite{gandhijacm}. Also, we should add that other authors in addition have edge capacities $c_e$ and then minimize the relative congestion, but it is easily seen that this just replaces the $C$ in the first type of constraints by $c_e C$. 

Since already the case of unit demands is NP-complete, optimal solutions seem difficult to obtain. The common solution concept is to (i) solve the linear relaxation of the ILP and obtain a fractionally optimal solution $(x^*,C^*)$; (ii) use \emph{path stripping} to decompose each flow $f_i$ encoded in $x^*_{\cdot}$ into a weighted sum $f_i = \sum_{P \in \PP_i} y^*_{iP} f_P$, where $\PP_i$ is a finite set of $s_i$--$t_i$ paths, for each $P \in \PP_i$, $f_P$ is the flow that has exactly one unit on each edge of $P$, and  $y^*_{iP} \in [0,1]$---note that all this implies $\sum_{P \in \PP_i} y^*_{iP} = r_i$; (iii) use randomized rounding to round all $y^*_{iP}$ to $y_{iP} \in \{0,1\}$ in such a way that the cardinality constraints $\sum_{P \in \PP_i} y_{iP} = r_i$ are maintained. Now $\sum_{P \in \PP_i} y_{iP} f_P$ is a flow from $s_i$ to $t_i$ with flow value $r_i$. These flows form a solution having a congestion of $C = \max_{e \in E} \sum_{i = 1}^k \sum_{P \in \PP_i, e \in P} y_{iP}$. Large deviation bounds show that this congestion is not far from the value $C^*$ given by the relaxation~\cite{raghthom,kleinbergthesis}:
\begin{align*}
	&C = O\bigg(\frac{\log m}{\log(2 \log m/C^*)}\bigg), \mbox{ if } C^* \le \log m;\\
	&C = C^* + O(\sqrt{C^* \log m}), \mbox{ if } C^* > \log m.
\end{align*}
Recall that $C^*$ is a lower bound for the optimal solution. Hence if $C^*$ is not too small compared to $m$, then this approach gives very good approximation factors.

\subsection{Algorithms Used}

To approximately or exactly solve our test instances, we used the following algorithms. Whenever running times permitted, we used the exact ILP-Solver ILOG CPLEX 11.0 to directly solve the ILP given by (\ref{ilpcongestion}) to (\ref{ilplast}). All other approaches involve solving the linear relaxation of the ILP (for which again we used CPLEX) and then different rounding methods. 

Since the ILP contains hard cardinality constraints, we cannot use the classical independent approach of Raghavan and Thompson (we did so, though, ignoring the cardinality constraints, to see if the cardinality constraints make rounding more difficult). There are two approaches to generate randomized roundings respecting cardinality constraints due to Srinivasan~\cite{sriniround} and the first author~\cite{ichstacs06}. Both can be derandomized~\cite{DW09,ichstacs06}, so that in total we have four rounding methods available. See the original papers or~\cite{DW09} for a more detailed discussion of these methods. All algorithms different from CPLEX were implemented in C/C++. 

\subsection{Experimental Set-up}

To analyze the questions discussed in the introduction, we regarded the following type of instances. Motivated by the fact that many routing problems have a two-dimensional flavor (e.g., the wire routing problem of~\cite{karprouting}), we chose a finite two-dimensional bi-directed grid. Note that this simple graph is far from trivial for routing, see, e.g., the thrilling one-turn routing conjecture in~\cite{karprouting}, which is, to the best of our knowledge still open.
 
We choose routing requests randomly as follows. Both $s_i$ and $t_i$ are chosen uniformly at random from $V$. To reduce otherwise the influence of randomness, we choose all demands in our main set of experiments as $r_i = 3$; in a second set of experiments we pick each~$r_i$ uniformly at random from~$\{1,\dots,5\}$.
We also tried placing the $s_i, t_i$ uniformly at random on the outer border of the grid, but saw no significant differences.

The size $n$ of the grid and the number of demands $k$ was varied to create different instance sizes and densities. All numerical values reported are the averages over at least $100$ runs. The times were measured on AMD dual processor 2.4 GHz Opteron machines.

\subsection{Analysis}

A summary of our results is presented in Tables~\ref{tab:congestion} and~\ref{tab:congestion_randdemand}. For the grid sizes $5 \times 5$, $10 \times 10$ and $15 \times 15$ together with demands ranging from $10$ to $75$, we state the average values of the congestion of an optimal solution (line 1 of each table) and the amounts by which a solution computed by one of the four randomized rounding approaches (lines 2, 3, 5, and 6) is worse (in percent). Lines 4 and 7 of the tables refer to an improvement discussed in the subsequent subsection. 

Particularly for instances that do show some congestion, we see that randomized rounding yields quite good solutions, much better than what the theoretical bounds would predict. Derandomization is clearly worth the small extra effort (cf. the running times in Table~\ref{tab:times}), reducing the gap to the optimum by roughly a half. 
The results with randomly chosen demands in Table~\ref{tab:congestion_randdemand} are qualitatively similar to those with uniform demands in Table~\ref{tab:congestion}, 
with perhaps a generally somewhat larger advantage for the derandomized algorithms.

Comparing the different methods, surprisingly, our experiments generally show that the bit-wise randomized rounding approach of~\cite{ichstacs06} (line 3) produced slightly better rounding errors than the tree-based one of~\cite{sriniround} (line 2). We do not understand this phenomenon currently. Among the derandomizations, as expected and similarly as for random instances~\cite{DW09}, the derandomization of the tree-based approach of~\cite{sriniround} given in~\cite{DW09} is superior to the derandomization of the bit-wise one in~\cite{ichstacs06}. This is due to the iterative nature of the latter, see again~\cite{DW09}.

We also used classical independent randomized rounding. Clearly, this does not produce feasible solutions in most cases. However, even ignoring this issue, we also observed that we typically have slightly larger congestions (e.g. in the sparse instance of $10$ demands in a $15\times 15$ grid, independent rounding lead to a congestion of $3.19$ compared to congestions of $2.80$ and $2.85$ for the randomized approaches of~\cite{ichstacs06} and~\cite{sriniround}).

Running times consumed by the different stages are mainly given in Table~\ref{tab:times}. All randomized rounding stages for each instance took less than 0.02 seconds. Less than a tenth of this is the time needed for the path-stripping in each instance. Hence these numbers are not given in the table. From the table, we see that the bit-wise derandomization takes about 20 times longer than the tree-based one, but both numbers are greatly dominated by the times for solving the LP (ignore the ``Heur.'' line for the moment). 

\begin{table}[t]%
\small
\centering
\subtable{
\begin{tabular}{l||r|r|r|r}
  $5 \times 5$ & 10 & 25 & 50 & 75 \\ \hline 
 Optimum & 3.37 & 6.21 & 10.98 & 14.76 \\
 RR \cite{sriniround}   & +9.23\% & +8.70\% & +7.29\% & +5.96\% \\
 RR \cite{ichstacs06}   & +7.13\% & +5.96\% & +5.01\% & +4.13\% \\
 RR+   &+6.35\% & +4.99\% & +4.28\% & +2.64\% \\
 DeRR \cite{ichstacs06} &+3.77\% & +3.54\% & +2.37\%& +1.90\% \\
 DeRR \cite{DW09} &+2.76\% & +1.93\% & +0.82\% & +0.88\% \\
 DeRR+ &+1.19\% & +1.13\% & +0.64\% & +0.47\% 
\end{tabular}
}
\subtable{
\begin{tabular}{l||r|r|r|r}
 $10 \times 10$  & 10 & 25 & 50 & 75 \\ \hline
 Optimum & 2.19 & 3.41 & 5.59 & 7.76 \\
 RR \cite{sriniround}   & +32.17\% & +39.88\% & +31.13\% & +24.48\% \\
 RR \cite{ichstacs06}   & +30.43\% & +36.07\% & +27.01\% & +19.07\% \\
 RR+   & +31.85\% & +24.63\% & +18.60\% & +12.50\% \\
 DeRR \cite{ichstacs06} & +22.33\% & +25.81\% & +15.56\% & +10.95\% \\
 DeRR \cite{DW09} & +16.38\% & +22.58\% & +11.63\% & +8.38\% \\
 DeRR+ &  +11.81\% & +10.85\% & +7.87\% &  +4.25\%
\end{tabular}
\label{tab:congestion10}
}
\subtable{
\begin{tabular}{l||r|r|r|r}
  $15 \times 15$ & 10 & 25 & 50 & 75 \\ \hline
 Optimum & 1.98 & 2.73 & 4.11 & 5.31 \\
 RR \cite{sriniround}   & +42.29\% & +58.97\% & +52.07\% & +45.70\% \\
 RR \cite{ichstacs06}   & +40.31\% & +59.34\% & +46.94\% & +41.40\% \\
 RR+   & +40.26\% & +45.79\% & +36.98\% & +28.49\% \\
 DeRR \cite{ichstacs06} & +23.42\% & +38.10\% & +31.63\% & +25.81\% \\
 DeRR \cite{DW09} & +13.76\% & +27.84\% &+21.90\% & +18.28\% \\
 DeRR+ & +16.59\% & +17.22\% & +18.49\% &+12.90\%
\end{tabular}
}
\normalsize
\caption{Congestions achieved by the 7 different approaches for grid sizes $5 \times 5$, $10 \times 10$ and $15 \times 15$. All demands are chosen as $r_i = 3$. The optimum was computed by solving the IP via CPLEX (not feasible for larger instances). For the other algorithms we state the relative increase of the congestion over the optimum.}
\label{tab:congestion}
\end{table}

\begin{table}
\small
\centering
\subtable{
\begin{tabular}{l||r|r|r|r}
  $5 \times 5$ & 10 & 25 & 50 & 75 \\ \hline
 Optimum & 3.79 & 6.65 & 11.08 & 15.08\\
 RR \cite{sriniround}   & +12.66\% & +9.47\% & +8.03\% & +6.03\% \\
 RR \cite{ichstacs06}   & +8.71\% & +6.92\% & +5.23\% & +4.51\% \\
 RR+   & +6.60\%& +4.36\% & +3.34\% & +2.98\% \\
 DeRR \cite{ichstacs06} & +2.64\%& +2.86\% & +2.53\%& +1.92\% \\
 DeRR \cite{DW09} & +1.58\% & +1.05\% & +0.99\% & +0.93\% \\
 DeRR+ & +0.79\% & +0.45\% & +0.63\% & +0.27\% 
\end{tabular}
}
\subtable{
\begin{tabular}{l||r|r|r|r}
  $10 \times 10$ & 10 & 25 & 50 & 75 \\ \hline
 Optimum & 2.60 & 3.70 & 5.86 & 7.94\\
 RR \cite{sriniround}   & +31.15\% & +39.73\% & +28.16\% & +23.93\% \\
 RR \cite{ichstacs06}   & +27.69\% & +32.70\% & +22.70\% & +18.64\% \\
 RR+   & +18.08\%& +24.05\% & +15.19\% & +14.11\% \\
 DeRR \cite{ichstacs06} & +16.92\%& +21.62\% & +14.51\%& +11.34\% \\
 DeRR \cite{DW09} & +8.85\% & +16.22\% & +10.58\% & +8.56\% \\
 DeRR+ & +6.15\% & +9.46\% & +6.14\% & +4.16\% 
\end{tabular}
}


\subtable{
\begin{tabular}{l||r|r|r|r}
  $15 \times 15$ & 10 & 25 & 50 & 75 \\ \hline
 Optimum & 2.19 & 2.99 & 4.29 & 5.71\\
 RR \cite{sriniround}   & +47.49\% & +55.52\% & +48.72\% & +40.28\% \\
 RR \cite{ichstacs06}   & +43.38\% & +47.49\% & +41.26\% & +36.43\% \\
 RR+   & +39.73\%& +36.12\% & +36.13\% & +26.09\% \\
 DeRR \cite{ichstacs06} & +26.48\%& +31.22\% & +27.27\%& +20.49\% \\
 DeRR \cite{DW09} & +19.18\% & +22.41\% & +22.14\% & +16.64\% \\
 DeRR+ & +17.81\% & +15.05\% & +16.78\% & +11.03\%
\end{tabular}
}
\caption{Congestions achieved as in Table~\ref{tab:congestion}, but with demands chosen independantly and uniformly at random from $r_i \in \{1, \dots, 5\}$. }
\label{tab:congestion_randdemand}
\end{table}


\begin{table}
\centering
\subtable{
\begin{tabular}{l||r|r|r|r}
 $5 \times 5$  & 10 & 25 & 50 & 75 \\ \hline
 IP (CPLEX)              & 0.0270 & 0.1076 &  0.343 & 0.776 \\
 LP (CPLEX)              & 0.0227 & 0.1078 &  0.301 & 0.697 \\
 Heur.                   & 0.0129 & 0.0380 &  0.065 & 0.096 \\
 DeRR \cite{DW09}        & 0.0009 & 0.0006 & 0.0015 & 0.0028\\
 DeRR \cite{ichstacs06}  & 0.0126 & 0.0270 & 0.0538 & 0.0589
\end{tabular}
}
\subtable{
\begin{tabular}{l||r|r|r|r}
 $10 \times 10$  & 10 & 25 & 50 & 75 \\ \hline
 IP (CPLEX)              & 0.3678 & 6.78 &  43.70 & 61.52 \\
 LP (CPLEX)              & 0.2349 & 4.71 &  32.03 & 34.02 \\
 Heur.                   & 0.6123 & 5.74 &  28.88 & 51.31 \\
 DeRR \cite{DW09}        & 0.0061 & 0.012 & 0.018 & 0.018\\
 DeRR \cite{ichstacs06}  & 0.1062 & 0.257 & 0.420 & 0.459
\end{tabular}
}
\subtable{
\begin{tabular}{l||r|r|r|r}
 $15 \times 15$  & 10 & 25 & 50 & 75 \\ \hline
 IP (CPLEX)              & 5.72 & 277.7 &  2057 & 7606 \\
 LP (CPLEX)              & 1.61 & 48.7 &  567 & 1135 \\
 Heur.                   & 9.42 & 131.6 &  819 & 2755 \\
 DeRR \cite{DW09}        & 0.026 & 0.044 & 0.063 & 0.07\\
 DeRR \cite{ichstacs06}  & 0.375 & 0.973 & 1.371 & 1.67
\end{tabular}
}
\caption{{Running times of the 5 algorithms in seconds. Given is the time for this particular step. For example, the running time of what is called ``DeRR+'' in Table~\ref{tab:congestion} is the sum of the values in lines ``LP'', ``DeRR'' and ``Heur.''.}}
\label{tab:times}
\end{table}


\subsection{A Heuristic Making Life Easier for Randomized Rounding}

As can be seen from the results presented so far, the different randomized rounding approaches usually find solutions that are not far from the optimum. We now propose and analyze a heuristic way to improve the performance. 

The rough idea is simple. Having solved the linear relaxation of the ILP, we know the optimal (relaxed) congestion $C^*$ that can be achieved. The congestion we end up with stems from this $C^*$ plus possible rounding errors inflicted in the congestion constraints (\ref{ilpcongestion}). It is clear that randomized rounding has a higher change to increase the congestion if there are many  congestion constraints satisfied with equality in the relaxation. 

Therefore, the heuristic we suggest is to resolve the LP with the following modifications. Let $\delta \in [0,C^*]$ be a parameter open for fine-tuning. We replace the congestion constraints (\ref{ilpcongestion}) by $\forall e \in E: \sum_{i = 1}^k x_{ie} \le C^*-\delta+z_e$, where $C^*$ is the (fixed) optimal congestion obtained from the first LP and $z_e \in [0,\delta]$ are new variables. The new objective is to minimize $\sum_{e \in E} z_e$. Since the $z_e$ are at most $\delta$, the flow given by a solution of this new LP also yields a congestion of at most $C^*$. However, the new objective punishes edges with total flow exceeding $C^*-\delta$. In consequence, the solution we obtain is also a solution for the original LP, but one that in addition tries to keep some room in the congestion constraints.

The experimental results are again presented in Table~\ref{tab:congestion}. Line~4 contains the results obtained by using randomized rounding as in~\cite{ichstacs06} after applying the heuristic and line~7 does so with the derandomization of~\cite{sriniround,DW09}. We did the same experiments with the other two rounding algorithms. Since the results were mainly inferior (to a similar extent as without the improvement), we omitted these numbers in the table. In all experiments, we chose $\delta=1$.

The results clearly show that using this heuristic can be worth the extra effort of solving a second LP. Apart from two instances with very small objective values 1.98 and 2.19, the heuristic always gains us a significant improvement. Surprisingly, these gains tend to be higher when using the derandomized rounding algorithm. 

It should be noted, though, that solving the second LP can be costly, as the numbers in Table~\ref{tab:times} indicate. 


\input{sec_maxcov_arXiv}

\end{document}

%% file: sec_maxcov_arXiv.tex
\section{Maximum Coverage: From Greedy and Rounding to Hybrid Approaches}

Another problem where dependent rounding has found application
is the Maximum Coverage problem. In this problem, the input is a set $\{S_1, \dots, S_n\}$ of sets and a budget
bound~$L$.  The task is to select a set of~$L$ sets to maximize the size of their union.
Additionally, there can be costs~$c_i$ associated with the sets, and weights or profits~$w_i$ associated
with the elements.  In this case the task is to maximize the weighted sum of the covered elements, 
subject to the constraint that the total cost of the sets is at most~$L$.

For the unit-cost case (when all set costs are equal to one), a $(1-1/e)$-approximation can be produced easily,
either through a greedy algorithm \cite{GerardCornuejols04011977} or via randomized rounding~\cite{sriniround,DBLP:journals/jco/AgeevS04};
see below for details.
This is also the best polynomial-time approximation ratio possible unless P=NP~\cite{Feige285059}.
For the general case (with weighted budgets), essentially the same ratio is possible using both approaches,
but some care has to be taken to handle sets of high cost~\cite{Khuller199939,sriniround}.

In this section, we report on experiments performed on max-coverage instances of various
types and from different sources, comparing the behavior of the greedy and rounding-based algorithms.
In addition, based on experiences from these experiments, we describe two forms of greedy/LP-rounding hybrid
algorithms, and observe significantly improved solution qualities. 
We also consider a budget-preserving rounding, improving on that used in~\cite{sriniround}
(a similar improvement can be found in the so-called weighted dependent rounding of~\cite{gandhifocs02}).
The algorithms, and our improvements to them, are described in Section~\ref{sec:mcmethods}.

One important special case we consider is max-domination in unit disk graphs, corresponding to covering
points in the plane under Euclidean distance.  For this problem, we develop a PTAS (polynomial-time approximation scheme)
as a further point of comparison.  
The PTAS and more information on the problem setting are given in Section~\ref{sec:planar}.
Some of our instances for this setting come from a real-world facility location problem~\cite{brazil_url},
referenced by the OR Library web page~\cite{orlib}; as a source of more benchmark instances,
we also convert facility location benchmarks, gathered at~\cite{UflLib}, to a max-coverage setting.
See Section~\ref{sec:mcsetup} for more information on our experimental setup.
Thereafter, the rest of the section contains experimental outcomes and conclusions.

\subsection{The Algorithms} \label{sec:mcmethods}

The \emph{greedy algorithm} repeatedly selects a set fitting in the budget that maximizes the ratio of the profit of the newly
covered elements to the cost of the set.  
In the unit-cost case, this is a~$(1-1/e)$-approximation~\cite{GerardCornuejols04011977};
with general costs, further modifications are needed, but for the case that the budget is large compared to the cost of the
most expensive set, as will be the case in our experiments, it is still a good approximation~\cite{Khuller199939}.

For the \emph{rounding-based approach}, we phrase the problem as an ILP as follows (let~$n$ be the number of sets, and~$m$ the number of elements in the instance).
\begin{align} 
&\max \sum_{i=1}^m w_ix_i \\
\mathrm{s.t.}\,      &  \sum_{i=1}^n c_iy_i \leq L \label{eqn:budget}\\
  &\forall i \in [m]: x_i \leq \sum_{j: i \in S_j} y_j \\
  &\forall i \in [m]: x_i \in [0,1] \\
  &\forall i \in [n]: y_i \in \{0,1\} 
\end{align}
Let us first focus exclusively on the unit-cost case. Getting a \emph{randomized} approximation for this case is simple.
Solve the linear relaxation of the above formulation, and let~$(x^*,y^*)$ be an optimal solution, of value~$W^*$.
Applying randomized rounding to~$y^*$ with a cardinality constraint preserving the sum~$\sum_i y_i^*$
now on expectation produces a $(1-1/e)$-approximation, due to the negative correlation properties of the rounding~\cite{sriniround}.

The \emph{derandomization} works via the method of conditional expectation. 
Consider the expected outcome of independently rounding the variables~$y^*$ above:
\begin{equation}
F(y^*) = \sum_{i=1}^m w_i(1-\prod_{j: i \in S_j} (1-y_j^*)).
\label{eqn:expmaxcov}
\end{equation}
As shown by Ageev and Sviridenko~\cite{DBLP:journals/jco/AgeevS04}, we can use~$F(y)$ directly as a guide
for the derandomization, and produce a rounding~$y\in \{0,1\}^n$ of~$y^*$, such that $F(y) \geq F(y^*)$.
As~$F(y^*) \geq (1-1/e)W^*$, the rounding~$y$ will be a~$(1-1/e)$-approximation.

As an additional alternative, we introduce \emph{gradient-based rounding}.  Recall that cardinality-preserving randomized rounding
works by repeatedly considering pairs of non-integral variables and readjusting their values, maintaining the sum, such that one of them becomes integral
(see e.g.~\cite{DW09}).
By gradient-based rounding, we attempt to identify the best pair of variables to select for adjustment in each step.
To truly find this pair would require~$O(n^2)$ comparisons, each with cost~$O(m)$, but we can approximate the selection by considering the gradient of~$F(y)$.  
It is easy to show that if~$y_i$ and~$y_j$ are non-integral, and if~$\frac {\partial F(y)} {\partial y_i} \geq \frac {\partial F(y)} {\partial y_j}$, then
moving mass from~$y_j$ to~$y_i$ will keep the value of~$F(y)$ non-decreasing.
Thus we only need to compute and update the partial derivatives~$\frac {\partial F(y)} {\partial y_i}$, which can be done analytically at a cost of~$O(nm)$ per step,
and we can in every step pair off the variables with largest and smallest values of partial derivative.  
While the total complexity of the rounding process becomes~$O(n^2m)$, as opposed to~$O(nm)$ for standard derandomized rounding, 
the time requirement is small in practice (see Section~\ref{sec:mcgenexp}).

Returning to the issue of weighted (knapsack) budget constraints,
Srinivasan gives a rounding procedure (Lemma~3.1 in~\cite{sriniround}) that approximately 
preserves the value of a weighted sum of the rounded variables, while guaranteeing negative correlation properties as in the unit-cost case.
However, the way in this is achieved for many settings causes infeasible running times.
To solve the problem, we consider a budget-preserving rounding procedure, as follows.
A similar rounding is found in the weighted dependent rounding used in~\cite{gandhifocs02}.

\begin{theorem} \label{thm:roundbudget}
Let~$y \in [0,1]^n$ such that~$\sum_i c_iy_i=L$.
In polynomial time, one can compute a rounding~$\tilde y \in \{0,1\}^n$ such that
$\sum_i c_i\tilde y_i \leq L+\max_i c_i$ and~$F(\tilde y) \geq F(y)$.
\end{theorem}

\begin{proof}
We refer to e.g.~\cite{DW09} or~\cite{DBLP:journals/jco/AgeevS04} for a description of the non-weighted, cardinality-preserving rounding procedure.
The only modification to this procedure is to replace the pair-rounding step, where variables~$y_i$ and~$y_j$ are adjusted so that one of them becomes integral.
Here, instead of keeping~$y_i+y_j$ constant, we instead maintain~$c_iy_i+c_jy_j$, that is, we change~$y_i$ and~$y_j$ to~$y_i-(\delta/c_i)$ and~$y_j+(\delta/c_j)$
for some appropriate~$\delta$. We will show that among any pair of adjustments in different directions, at least one keeps~$F(y)$ non-decreasing.

Let~$F(y)$ be defined as~(\ref{eqn:expmaxcov}), and let two non-integral variables~$y_i$ and~$y_j$ be selected,~$i<j$, 
and define
\[
g_{a}(\delta) = F(y_1, \dots, y_{i-1}, y_i+\delta, \dots, y_j-a\delta, \dots, y_n).
\]
We show that for any~$a>0$,~$g(\delta)$ is a convex function.  Indeed, the only terms of~$g(\delta)$ that are not constant or linear in~$\delta$
correspond to elements~$x \in S_i \cap S_j$:
\[
w_x (1-(1-y_i-\delta)(1-y_j+a\delta)\prod_{k: x \in S_k, k \notin \{i,j\}} (1-y_k)
\]
where the only non-linear term is~$c \cdot \delta^2$, for some constant~$c>0$.

Thus one of the values~$g(-\delta_0)$ and~$g(\delta_1)$, for~$\delta_i\geq 0$ are at least as large as~$g(0)=F(y)$,
and the rounding search as in~\cite{DBLP:journals/jco/AgeevS04} can be performed.

Finally, the~$\max_i c_i$ term occurs because we may end up with one single non-integral variable at the end of the process.
\end{proof}
As the first paragraph of this proof shows, to combine the gradient-based rounding with this rounding process
we only have to divide each component~$i$ of the gradient by the corresponding cost~$c_i$ before variable selection.

\subsection{A Unit Disk Maximum Domination PTAS} \label{sec:planar}

As a source of real-world instances, we consider a type of max-coverage instances derived from planar point set data.  
Given a set of points $P=\{p_1, \dots, p_n\}$ in the plane and a \emph{diameter} $d$, we define a graph (the unit disk intersection graph) by letting two points $p_i$, $p_j$ be connected if and only if the Euclidean distance between them is at most $d$.
In this graph, we consider the problem of max-domination, where selecting a vertex~$v$ covers~$v$ and all its neighbors. 
Interpreted as max-coverage, we thus get an instance where every vertex corresponds to one set and one element.
All sets will have unit cost; the elements may have weights, interpreted as the profit of covering them.

This problem is NP-hard, as follows from the hardness of Minimum Dominating Set in unit disk graphs \cite{Masuyama198157}.
However, it has good approximation properties---using tools known from the literature (e.g.~\cite{Glasser20051}), we are able to 
provide a polynomial-time approximation scheme (PTAS).  


The PTAS follows the grid-based shifting strategy of Arora~\cite{AroraPTAS}, as also applied to a related problem
on the placement of wireless base stations by Gla{\ss}er et al.~\cite{Glasser20051}.
We show the result for instances with unit costs and arbitrary profits, as this is closest to our instances,
and having both costs and profits arbitrary makes the problem as hard as max-coverage.

\begin{theorem} \label{th:ptas}
For any~$\ell>1$, the Max Domination problem on unit disk graphs with weighted vertices and unit costs
admits a~$(1-2/\ell)$-approximation in time~$n^{O(\ell^2)}$.
\end{theorem}
\begin{proof}
Assume w.l.o.g. that the diameters of the disks is one,
and  divide the space into a regular grid with unit sides, 
so that every point resides in one grid box.  
Repeat the following for all values~$\ell_h,\ell_v \in \{0, \dots, \ell-1\}$.

Mark every~$\ell$:th column, starting from number~$\ell_h$, and 
every~$\ell$:th row, starting from number~$\ell_v$.  Any point
which is a member of a marked row or column gets payoff value~$0$.
Now for every subgrid of side~$l+1$, framed by marked rows and columns,
enumerate all optimal solutions as functions from the budget used to
the payoff achieved.  

Since our concern is the unit-cost case, the number of such payoffs is 
bounded by the cardinality of the solution which dominates all points,
i.e., by the size of a minimum dominating set.  This in turn is bounded
by the independence number, since a maximal independent set is dominating,
which is~$O(l^2)$, since each disk covers an area of~$O(1)$ and the subgrid
has total area~$O(l^2)$.  

Thus, we can in polynomial time create a vector of the~$O(l^2)$ different
payoff values that are received by investing a certain number~$t$ of sets
in the current subgrid; for each such value~$t$, the optimal solution
can be found in~$n^{t+O(1)}$ time by enumeration.  Thus in~$n^{O(l^2)}$ time
we can get the budget:payoff vectors for every framed subgrid.

Thereafter, we can collect these vectors into a kind of knapsack problem,
which can be solved by dynamic programming over subgrids and budget,
i.e.,~$f(i,t)$ being the best possible payoff for using~$t$ of the budget
to cover vertices in the first~$i$ subgrids.
Note that unmarked points in different framed subgrids must be covered by 
different points, and vice versa, an unmarked point can only be covered
by a point in its framed subgrid.  Thus a solution to the instance, after
the marking step has been performed, decomposes into one solution per
framed subgrid, and we are able to compute this optimally in polynomial 
time.

Finally, every point is marked at exactly one value of~$\ell_h$, 
thus for some shift value~$\ell_h$ at most a fraction~$1/\ell$ of the 
points are marked; by the same argument, for some value of~$\ell_v$,
at most a fraction~$1/\ell$ of the remaining points are marked.
This means that there is a pair of shift values~$\ell_h,\ell_v$ for which
the computed solution represents at least a fraction~$(1-1/\ell)^2\geq (1-2/\ell)$ of the optimal solution.
\end{proof}

%
%
%
%

In our implementation, we replace certain steps by an MIP solver, specifically 
the exhaustive enumeration phase for the subproblems, and the dynamic programming knapsack step.
With these modifications, execution of the algorithm for non-trivial approximation ratios becomes possible.
Unfortunately, we will see that even with these modifications, the PTAS approach is inferior to the greedy 
and rounding algorithms for realistic approximation settings.

\subsection{Experimental Setup} \label{sec:mcsetup}

We perform experiments using the greedy algorithm and different combinations of LP-relaxation and rounding procedure.

For the rounding, we use three variants.
Recall that the expected value of a single randomized rounding equals~$F(y^*)$, and can thus 
(unlike in Section~\ref{sec:routing}) be computed exactly. We consider three ways of boosting this value.
The first is to simply apply randomized rounding $1000$ times and pick the best result;
the second is derandomization, using Srinivasan-type rounding directly on~$F(y)$,
with arbitrary order of variable comparison.
The third is the gradient-based rounding of Section~\ref{sec:mcmethods}.

For problems with a weighted budget, we use the same three rounding methods, where the best-of-1000 and the derandomized roundings
use Lemma~3.1 of~\cite{sriniround}, and the gradient-based rounding uses Theorem~\ref{thm:roundbudget}.
Should a rounding exceed the budget constraint, we will greedily discard sets until the budget bound is reached.

We use CPLEX for both LP- and ILP-solving; all other algorithms were implemented by the authors in C.

Our instances are of two types.  For the unit disk max-domination problem, we use benchmark instances stemming from a real-world facility location problem, previously used in~\cite{brazil_use} and available at~\cite{brazil_url}.
For these instances, a demand is provided with every point; we use these demands as profit values.
To complement this, we use instances converted from facility location problems, in most cases downloaded from UflLib~\cite{UflLib}. 
In both cases, we convert the instances to max-coverage by selecting an appropriate distance threshold for membership.
In the case of the M* instances~\cite{mstar}, the distances were pre-scaled by demand values, making the distances inappropriate for our use; we remove this scaling, and use the demands as profit values instead.


\subsection{Experiments} \label{sec:mcgenexp}

\begin{table}[b]
\begin{center}
\begin{tabular}{c|c||c|c|c|c|c|c}
Instance & $n$ & LP & Random-1000 & Derand. & Grad. rounding & Greedy  & IP \\
\hline
br818-400-30 &818 & $0.36$s & $0.80$s & $0.11$s & $0.09$s & $0.33$s & 25s \\
kmed1-1k-37  & 1000 & $0.97$s & $1.46$s& $0.22$s & $0.25$s & $0.90$s & $>3600$s \\
MR1-060-16.5 & 500 & $0.36$s & $0.74$s& $0.05$s&  $0.04$s & $0.14$s & $>3600$s\\
\end{tabular}
\caption{Running times for various instances and algorithms; the times for the rounding methods exclude the time for solving the relaxation.}
\label{tab:mctimes}
\end{center}
\end{table}

We now report the results of our experiments. 

To begin with, we show the running times of the different methods on various instances in Table~\ref{tab:mctimes}.
Note that the LP solver once again uses a significant fraction of our running time,
and that performing many random roundings becomes more costly than derandomization, due to the need to evaluate the objective value for each solution.
The low numbers for the gradient-based rounding, as compared to the derandomization, can partly be explained by the gradient-based rounding being problem-specific,
while the derandomization uses general-purpose code.

\begin{table}[bt]
\begin{center}
\begin{tabular}{l|r|r||r|r r r r|r}
Name         & Size  &Budget& Greedy & LP: once           & 1000 &    derand  &     gradient      & Optimum \\
\hline
Chessboard   & $144$ &16    & $130$  &    $144$           & $144$&    $144$   &    $144$        & 144 \\ 
FPP ($k=17$) & $307$ &17    & $290$  &    $200$           & $210$&    $230$   &    $290$        & 290 \\ 

br818-400    & 818   & 30   &$28054$ &    $22157$        &$26199$&  $27397$   & $28448$         & 28709 \\ 

kmed1-1k      & 1000  & 37   & $948$  &    $709$           &$817$ &      $923$ & $962$           & 993--95 \\ 
MR1-060      & 500   &16500 &$1444$  &   $1179$           &$1254$&  $1402$    & $1445$          & 1462--94 \\

%
%
%
%
%
%
\end{tabular}
\end{center}
\caption{Experiments on single instances.  The data is averaged over 100 runs where the input data is randomly permuted.
The column ``LP: once'' shows the expected outcome of a single randomized rounding; the following three columns show our three rounding methods.  
The optimum gives the best upper and lower bound achieved by an IP solver after one hour of running time.
}
\label{tab:noweight}
\end{table}


Table~\ref{tab:noweight} shows results for some individual instances (described below).
The first we want to highlight are the Chessboard and Finite Projective Plane (FPP) instances.
These classes, proposed in~\cite{MYSTERY}, downloaded from \cite{UflLib},
are not intended as examples of realistic real-world problems, but rather serve to 
reveal the differences between the different approximation approaches. 
The Chessboard instances are instances on a chessboard with side~$3k$,
where the elements are the squares, and the sets are $3 \times 3$ subgrids (the set of squares reachable by a king within one step).
For us, these instances serve as a test for whether an algorithm discovers an optimal tiling. 
The FPP instances are graphs of a regular degree but more complicated structure.
Since all instances of these classes have equivalent combinatorial structure, we use only one instance per class in our experiments.
In both cases, the budget is set at the most difficult setting, which turns out to be just where the LP-relaxation can cover all or almost all elements.

The results immediately show the reason to pursue hybrid greedy/rounding algorithms.  
For the Chessboard instance, the LP-optimum is already integral, and thus every LP-rounding-based algorithm discovers an optimal solution.
On the other hand, these instances are difficult for the greedy algorithm, as a few early mistakes, when all sets seem equivalent, will hurt the end tiling.
In the FPP instances, however, we see the opposite effect. 
Here, upon inspection, we find that the LP-optimum is a useless mix of taking an equal, small amount of almost every variable,
leaving all the work of finding a good integral solution up to the rounding. 
On the other hand, the greedy algorithm performs very well here; runs with an ILP-solver show that it is at most one step away from the optimum.
We find that our proposed hybrid, the gradient-based rounding, produces consistent top-quality results in both test cases.

\begin{table}[b]
\begin{center}
\begin{tabular}{c|cr|cr|cr|c|cc}
Instance & \multicolumn{2}{|c|}{PTAS ($\ell=3)$} & \multicolumn{2}{|c|}{PTAS ($\ell=5$)} & \multicolumn{2}{|c|}{PTAS ($\ell=7$)} & Greedy & \multicolumn{2}{|c}{IP}    \\
         &       Value & Time                 &  Value & Time                      & Value & Time                  & Value  & Value & Time \\
\hline
br818-400-20 &   20742 & $1.3$s               & 22857  &  9s                       &  24129      &  69s            & 25247  & 26192 & $0.5$s\\
br818-400-25 &   23008 & $1.2$s               & 24683  &  9s                       &  25308      &  71s            & 26907  & 27670 & 9s\\
br818-400-30 &   23951 & $1.3$s               & 24909  & 10s                       &  27270      &  74s            & 28054  & 28709 & 25s\\
\end{tabular}
\end{center}
\caption{PTAS performance compared to the greedy algorithm and IP solver.}
\label{tab:ptas}
\end{table}

\begin{figure}[tb]
\begin{center}
\includegraphics{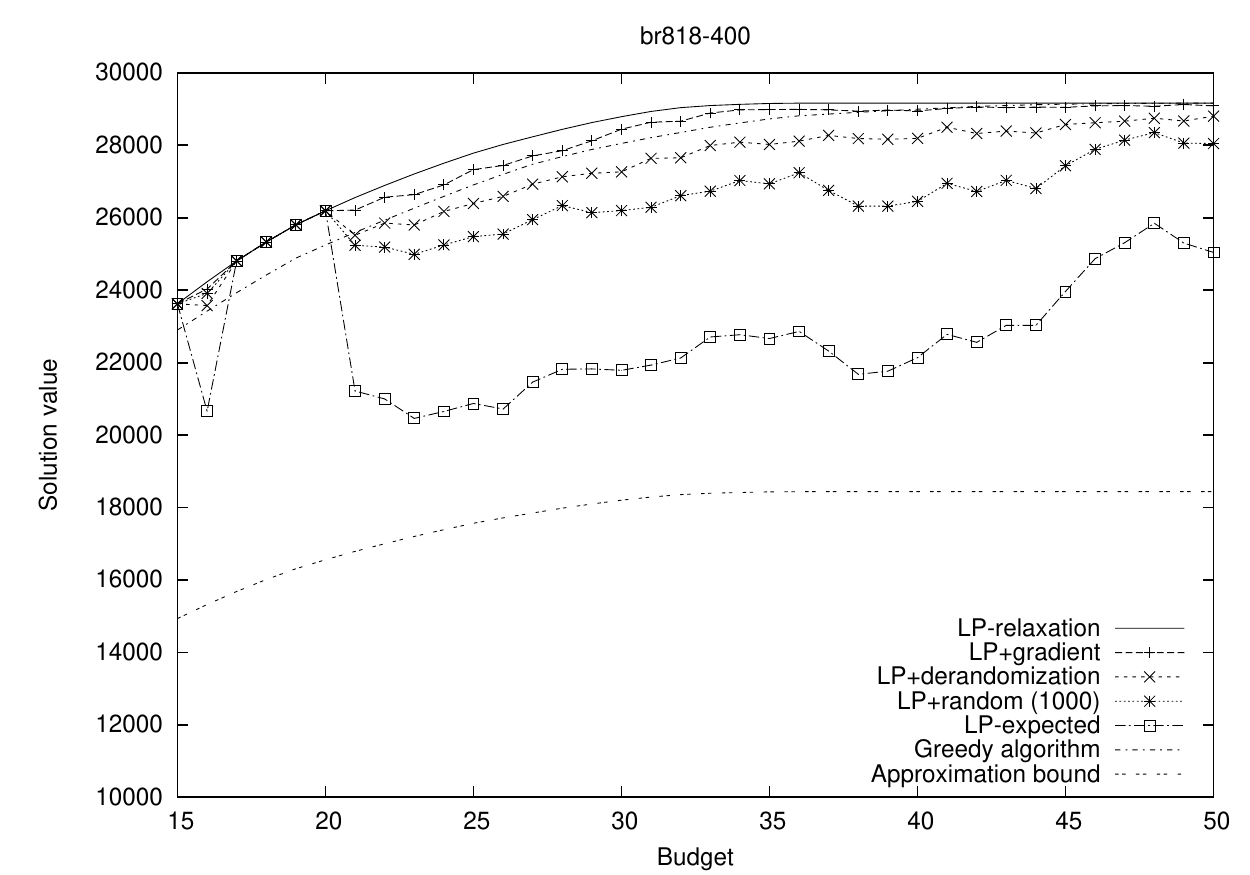}
\end{center}
\caption{Results for unit disk max-coverage instance br818-400. 
The plot shows the value of the LP-relaxation, the outcome of the three rounding methods, and the expected value of a single rounding against the greedy algorithm.  The approximation bound shows $(1-1/e)$ times the LP optimum.}
\label{fig:planar}
\end{figure}

\begin{figure}
\begin{center}
\includegraphics{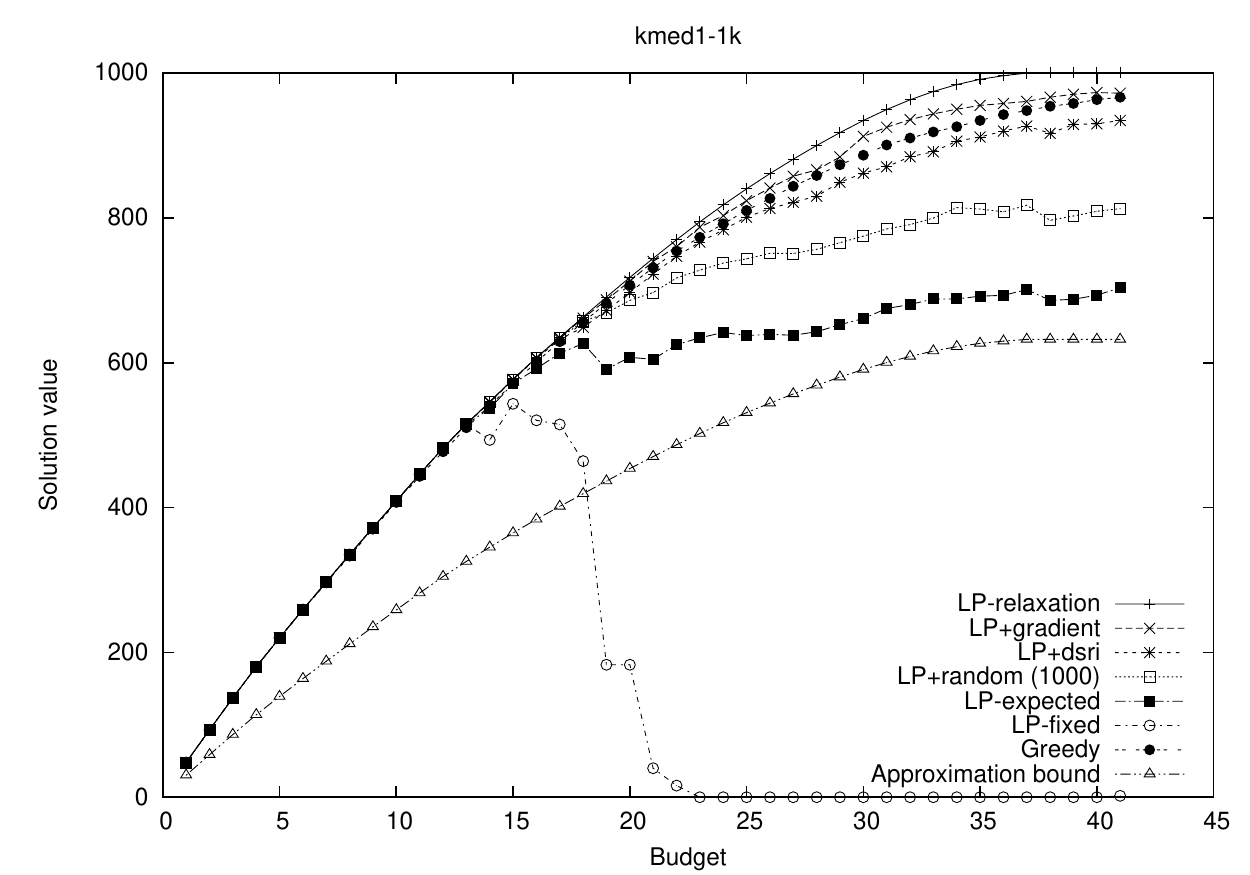}
\end{center}
\caption{Experimental outcome for instance kmed1 with distance threshold $1000$. 
The plot shows the value of the LP-relaxation, the outcome of the three rounding methods, and the expected value of a single rounding against the greedy algorithm.  The approximation bound shows $(1-1/e)$ times the LP optimum.}
\label{fig:kmed}
\end{figure}

\begin{figure}
\begin{center}
\includegraphics{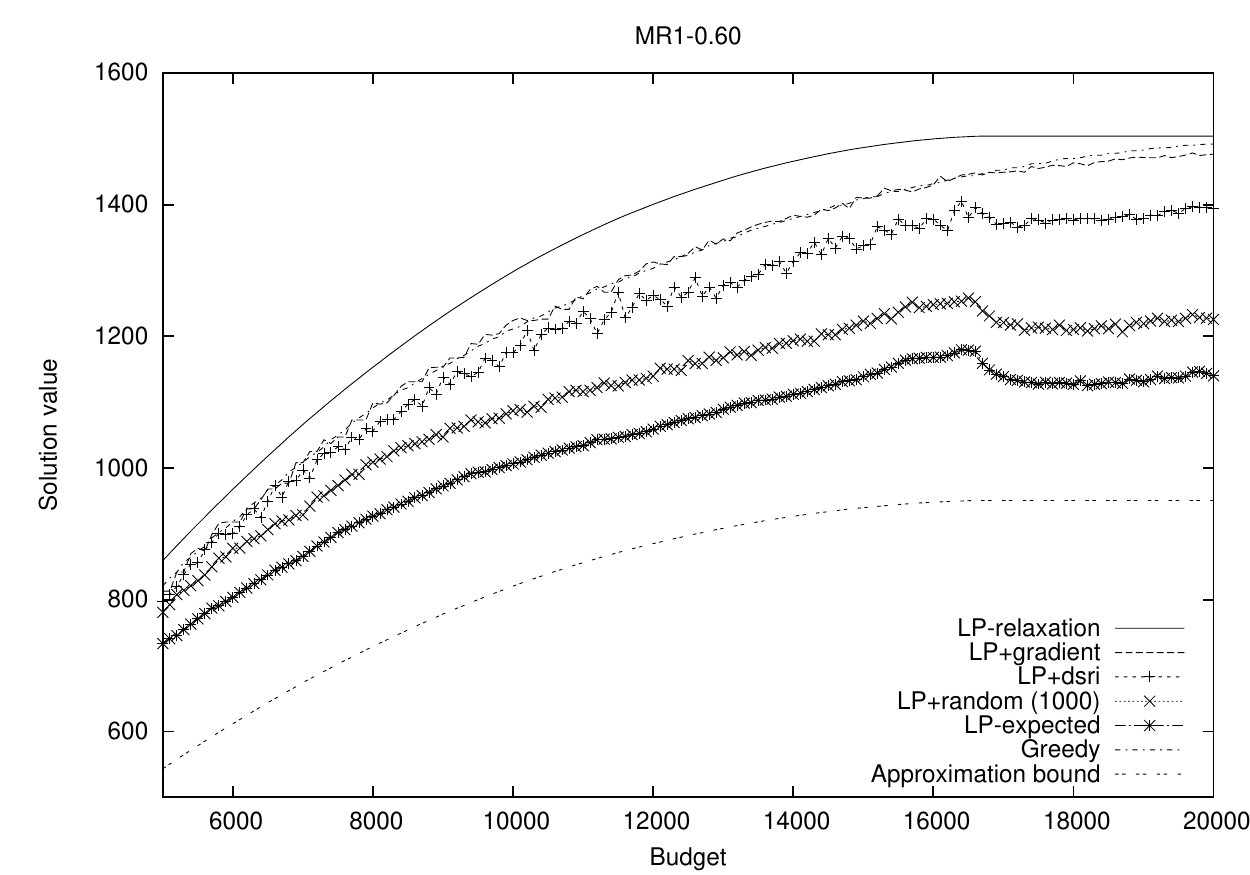}
\end{center}
\caption{Experimental outcome for instance MR1 with distance threshold $0.6$. 
The plot shows the value of the LP-relaxation, the outcome of the three rounding methods, and the expected value of a single rounding against the greedy algorithm.  The approximation bound shows $(1-1/e)$ times the LP optimum.}
\label{fig:budget}
\end{figure}

We now focus on the unit disk max-domination problem, with instances as described in Section~\ref{sec:planar}.
We select the largest instance, with~$818$ points inscribed in a box of sides~$6395$ by~$3975$, and use a distance threshold of~$400$.
This was chosen as a good balance, as too small or too large values (e.g., 100 resp.\ 800) creates too simple instances.
Figure~\ref{fig:planar} shows the behavior of the main algorithms (excluding the PTAS) for this instance, as depending on the budget.
Observe that the LP-rounding approach is very powerful for small budgets (up to 20), while further guidance is needed for larger budgets.  The gradient-based LP-rounding, providing just such guidance, produces top values throughout, frequently better than either the greedy or the standard rounding algorithms.
This instance also appears in Tables~\ref{tab:noweight} and~\ref{tab:mctimes} under the name br818-400 or br818-400-$L$, where~$L$ is the budget bound.
Another instance class of the same type is the \emph{k-median} instances~\cite{UflLib}.  Here we use the one named 1000-10, with a threshold of~$1000$, 
occurring in Tables~\ref{tab:noweight} and~\ref{tab:mctimes} as kmed1-1k or kmed1-1k-$L$, with Figure~\ref{fig:kmed} displaying the same data as Figure~\ref{fig:planar}.
In addition, Figure~\ref{fig:kmed} contains a plot of the fraction of the solution weight covered by integrally chosen sets (the line ``fixed'' in the figure).
The figure confirms that the LP-rounding approach is powerful for small budgets.

For concerns of clutter, the PTAS is not included in the figures, but its data is given separately in Table~\ref{tab:ptas}.
Note that for every feasible setting, the PTAS is both of lower quality and significantly slower than the alternatives.
The k-median-instance is omitted, since we lack point data for it.


Finally, we also consider an instance class with weighted budgets, namely the M* instances proposed in~\cite{mstar}, again downloaded from~\cite{UflLib}.
These are instances with uniformly random distances (i.e., random membership after conversion to max-coverage), but with facility costs (i.e., set costs)
chosen so that facilities close to many customers are more expensive.  The authors of~\cite{mstar} propose that such cost structures would arise in some real-world situations.
Table~\ref{tab:noweight} and Figure~\ref{fig:budget} give the results for instance MR1 with distance threshold~$0.6$, under the name MR1-060.
In general for this class, we found that the greedy algorithm and the gradient-based rounding method produce practically identical results, 
while the other methods are inferior to this.

\subsection{A Greedy/LP Hybrid.}

\begin{table}
\begin{center}
\begin{tabular}{l|r|r|r|r|r}
      & \multicolumn{5}{c}{Rounding Algorithm}\\
Relaxation & Integral part &  Expectation  &Random-1000 &  Derand & Gradient \\
\hline
LP           &$  4646$ & $21871$ &$ 26251$ &$ 27352$ &$ 28447$ \\
Hybrid $0.1$ &$ 12635$ & $24936$ &$ 27500$ &$ 27865$ &$ 28328$ \\
Hybrid $0.2$ &$ 22921$ & $27189$ &$ 28365$ &$ 28141$ &$ 28476$ \\
Hybrid $0.3$ &$ 28188$ & $28304$ &$ 28339$ &$ 28338$ &$ 28339$ \\
Hybrid $0.4$ &$ 28215$ & $28215$ &$ 28215$ &$ 28215$ &$ 28215$ \\
Hybrid $0.5$ &$ 27817$ & $26837$ &$ 28050$ &$ 28019$ &$	28082$ \\
Hybrid $0.6$ &$ 27970$ & $27530$ &$ 28048$ &$ 28048$ &$	28058$ \\
Hybrid $0.7$ &$ 28058$ & $28058$ &$ 28058$ &$ 28058$ &$ 28058$ \\
Hybrid $0.8$ &$ 28054$ & $28054$ &$ 28054$ &$ 28054$ &$ 28054$ \\
Hybrid $0.9$ &$ 28054$ & $28054$ &$ 28054$ &$ 28054$ &$ 28054$ \\
\end{tabular}
\caption{Results for combining greedy pre-selection with randomized rounding, averaged over 100 runs, complete results.
The data is for the instance br818-400-30, where the greedy algorithm alone produces value $28054$.}
\label{tab:oneinstancebrfull}
\end{center}
\end{table}

\begin{table}
\begin{center}
\begin{tabular}{l|r|r|r|r|r}
      & \multicolumn{5}{c}{Rounding Algorithm}\\
Relaxation & Integral part &  Expectation  &Random-1000 &  Derand & Gradient \\
\hline
LP           & $0.17$	& $700$  & $816$   & $	923$   	& $961$  \\
Hybrid $0.1$ & $137$	& $755$  & $857$   & $	940$   	& $970$  \\
Hybrid $0.2$ & $312$   	& $817$  & $904$   & $	957$    & $974$  \\
Hybrid $0.3$ & $757$   	& $925$  & $962$   & $	966$   	& $972$  \\
Hybrid $0.4$ & $797$   	& $934$  & $960$   & $	959$   	& $964$  \\
Hybrid $0.5$ & $915$   	& $947$  & $954$   & $ 	954$   	& $955$  \\
Hybrid $0.6$ & $929$   	& $948$  & $952$   & $	952$  	& $953$  \\
Hybrid $0.7$ & $931$   	& $947$  & $950$   & $	951$   	& $950$  \\
Hybrid $0.8$ & $947$	& $948$  & $948$   & $	948$   	& $949$  \\
Hybrid $0.9$ & $948$   	& $948$  & $948$   & $	948$   	& $948$  \\
\end{tabular}
\caption{Results for combining greedy pre-selection with randomized rounding, averaged over 100 runs, complete results.
The data is for the instance kmed1-1k-37, where the greedy algorithm alone produces value $948$.}
\label{tab:oneinstancekmfull}
\end{center}
\end{table}

\begin{table}
\begin{center}
\begin{tabular}{l|r|r|r|r|r}
      & \multicolumn{5}{c}{Rounding Algorithm}\\
Relaxation & Integral part &  Expectation  &Random-1000 &  Derand & Gradient \\
\hline
LP           &$     0$     & $ 1179$       &$ 1257 $    &$ 1381 $ &$  1446$ \\
Hybrid $0.1$ &$   369$     & $ 1232$       &$ 1301 $    &$ 1402 $ &$  1450$ \\
Hybrid $0.2$ &$   624$     & $ 1285$       &$ 1344 $    &$ 1421 $ &$  1443$ \\
Hybrid $0.3$ &$   821$     & $ 1316$       &$ 1371 $    &$ 1413 $ &$  1447$ \\
Hybrid $0.4$ &$   974$     & $ 1349$       &$ 1391 $    &$ 1432 $ &$  1448$ \\
Hybrid $0.5$ &$  1109$     & $ 1381$	   &$ 1413 $    &$ 1436	$ &$  1450$ \\
Hybrid $0.6$ &$  1226$     & $ 1408$	   &$ 1432 $    &$ 1438 $ &$  1441$ \\
Hybrid $0.7$ &$  1355$     & $ 1431$	   &$ 1443 $    &$ 1440 $ &$  1440$ \\
Hybrid $0.8$ &$  1440$     & $ 1446$	   &$ 1440 $    &$ 1445 $ &$  1440$ \\
Hybrid $0.9$ &$  1443$	   & $ 1443$	   &$ 1443 $    &$ 1442 $ &$  1443$ \\
\end{tabular}
\caption{Average value and standard deviation for all ways to relax-and-round one instance (MR1, distance $0.6$, budget~$16500$).
The pure greedy algorithm gives 1444. LP bound is 1503, optimum between $1462$ and $1494$.}
\label{tab:oneinstancemr}
\end{center}
\end{table}


Motivated by our results, we consider a different, more general form of greedy/LP hybrid than the gradient rounding.
Before we commence with the LP-rounding, we allocate some portion of the budget to greedy pre-selection,
and apply the LP-relaxation and rounding using the remaining budget to the thus reduced problem.
In Tables~\ref{tab:oneinstancebrfull}--\ref{tab:oneinstancemr} we examine the performance of such a hybrid on the 
br818-400-30, kmed-1k-37, resp.\ MR1-060-16.5 instances.  
We see that such a hybrid algorithm with a carefully chosen threshold can produce results superior to either the greedy or the rounding algorithm on their own.
The benefits for the gradient-rounding approach seem less consistent, although improvement is visible for the kmed1-1k instance.
We further report that with a pre-selection fraction of~$0.3$, both the Chessboard and the FPP instances of Table~\ref{tab:noweight} receive optimal solutions.